\documentclass[twoside]{article}

\usepackage{url,intmacros,graphicx}
\usepackage{amsmath,amssymb,amsthm}
\usepackage{color}
\usepackage{cases}
\usepackage{amsfonts,amsmath}
\usepackage{enumerate}
\usepackage{paralist}
\usepackage{hyphenat}
\usepackage[utf8]{inputenc}

\def\BibTeX{{\rm B\kern-.05em{\sc i\kern-.025em b}\kern-.08em
    T\kern-.1667em\lower.7ex\hbox{E}\kern-.125emX}}

\newcommand{\email}[1]{\\ \small{\url{#1}} \\}
\newcommand{\institution}[1]{\\ \parbox{3.0in}{\small{#1}}}
\newcommand{\keywords}[1]{\small\textbf{Keywords: }#1}
\newcommand{\AMSsubj}[1]{\noindent\textbf{AMS subject classifications:
}#1}
\newcommand\whenaccepted{Submitted: ???;
                         Revised: ????;
                         Accepted: ????;
                         Posted: ????}

\def\qed{\unskip\kern10pt{\unitlength1pt\linethickness{.4pt}\framebox(6,6){}}}

\newtheorem{definition}{Definition}
\newtheorem{theorem}{Theorem}

\newtheorem{lemma}[theorem]{Lemma}

\newtheorem{remark}{Remark}

\newenvironment{glists}[4]{
\begin{list}{}{
\setlength{\labelwidth}{#2}
\setlength{\labelsep}{#3}
\setlength{\leftmargin}{#1}
\addtolength{\leftmargin}{\labelwidth}
\addtolength{\leftmargin}{\labelsep}
\setlength{\parsep}{#4}
\setlength{\topsep}{\parsep}
\setlength{\itemsep}{\parsep}
\setlength{\listparindent}{0in}
}
}{
\end{list}
}
\newcommand{\iteml}[1]{\item[#1 \hfill]}

\title{The Secant-Newton Map is Optimal \\Among Contracting $n^{th}$ Degree Maps \\for $n^{th}$ Root Computation
\footnote{\whenaccepted}
}

\author{Kayla Bishop 
\institution{Department of Mathematics, North Carolina State University,
Box 8205, Raleigh NC 27695, USA}
\email{kbishop2@ncsu.edu}
\and Hoon Hong
\thanks{This research was partly supported by US National Science Foundation Grant 1319632.}
\institution{Department of Mathematics, North Carolina State University,
Box 8205, Raleigh NC 27695, USA}
\email{hong@ncsu.edu}
}

\date{}

\begin{document}

\maketitle

\begin{abstract}
Consider the problem: given a real number $x$ and an error bound $\varepsilon
$, find an interval such that it contains $\sqrt[n]{x}$ and its width is less
than~$\varepsilon$. One way to solve the problem is to start with an initial
interval and to repeatedly update it by applying an interval refinement map on
it until it becomes narrow enough. In this paper, we prove that the well known
Secant-Newton map is  \emph{optimal} among a certain family of natural generalizations.
\end{abstract}

\noindent\keywords $n^{th}$ root, interval mathematics, secant, Newton,
contracting

\noindent \AMSsubj 65G20, 65G30

\section{Introduction}
Computing the $n^{th}$ root of a given real
number is a fundamental operation.\footnote{
This paper is a sequel to \cite{erascu_square_2013} where square-root ($n=2$) was considered. In this paper, we generalize the result to the $n^{th}$ root.  For the readers who have not yet read \cite{erascu_square_2013}, we will make this paper self-contained. 
For the readers who have read \cite{erascu_square_2013}, we will do our best to use the same/similar notations and structuring/styles so that the readers can easily identify the similarities and the differences between the two papers.
} 
Naturally, various numerical methods have
been developed
\cite{wensley_class_1959,%
moore_interval_1966,%
moore_introduction_2009,%
hart_computer_1968,%
cody_software_1980,%
alefeld_introduction_1983,%
Revol_interval_2003,%
morrison_method_1956,%
meggitt_pseudo_1962,%
beebe_accurate_1991,%
fowler_square_1998,%
bruce_cube_1980,%
dubeau_nth_newton_2009,%
hernandez_nth_alg_2004,%
laufer_iteration_1963}.
In this paper, we consider an interval version of the
problem~\cite{moore_interval_1966,alefeld_introduction_1983,moore_introduction_2009}%
: given a real number~$x$ and an error bound $\varepsilon$, find an
interval
such that it contains $\sqrt[n]{x}$ and its width is less
than~$\varepsilon$.
One way to solve the problem starts with an initial
interval and repeatedly updates it by applying a \emph{refinement} map on it until it becomes narrow enough (see below).
\smallskip
\begin{glists}{2.4em}{2em}{0em}{0.1em}
\iteml{\sf in:}   $x > 0,\;\;\varepsilon > 0$
\iteml{\sf out:}  $I$, interval such that $\sqrt[n]{x} \in I$ and {\sf
width}$(I) \leq \varepsilon$
\end{glists}
\smallskip
\begin{glists}{2em}{0em}{0em}{0.1em}
\iteml{} $I \leftarrow [\min(1,x),\max(1,x)]$
\iteml{} \sf{while width}$(I) > \varepsilon$
\iteml{} $\;\;\;\;\;\;\;I \leftarrow R(I,x)$
\iteml{} \sf{return} $I$
\end{glists}
\smallskip
\noindent A well known refinement map $R^{*},$ tailored for $n^{th}$ root computation, is
obtained by combining the secant map and the Newton map where the
secant/Newton map is used for determining the lower/upper bound of the
refined
interval, that is,
\[
R^{*}:\;\;[L,U],x\;\;\mapsto\;\;\left[  L+\frac{x-L^{n}}{L^{n-1}+L^{n-2}U+...+LU^{n-2}+U^{n-1}},U+\frac{x-U^{n}}{nU^{n-1}}\right]
\]
which can be easily derived from Figure~\ref{SN}.
\begin{figure}[th]
\center  \includegraphics[width=0.5\textwidth]{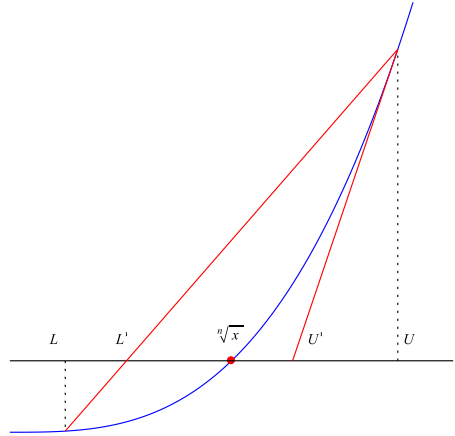} \label{SN}%
\caption{Derivation of Secant-Newton map}%
\end{figure}

A question naturally arises. \emph{Is there any refinement map which is
better
than Secant-Newton?} In order to answer the question rigorously, one
first
needs to fix a search space, that is, a family of maps in which we
search
for
a better map. In this paper, we will consider the family of all the
\textquotedblleft natural generalizations\textquotedblright\ of the
Secant-Newton
map. The above picture shows that the Secant-Newton map is contracting, that
is,
$L\leq L^{\prime}\leq\sqrt[n]{x}\leq U^{\prime}\leq U$. Furthermore, it
\textquotedblleft scales properly,\textquotedblright\  that is, if we
multiply\ $\sqrt[n]{x},\ L\ $and $U\ $by a number, say $s,\ $ then
$L^{\prime}$
and $U^{\prime}$ are also multiplied by $s.$ This is due to the fact
that
the
numerators are $n^{th}$ degree forms in $\sqrt[n]{x},\ L\ $and $U\ $and the
denominators are $n-1^{th}$ degree forms. These observations suggest the following
choice
of a search space: the family of all the maps with the form%
\begin{align*}
R &  :\;[L,U],x\mapsto\lbrack L^{\prime},U^{\prime}]\\
L^{\prime} &  =L+\frac{x+p_{0}L^{n}+p_{1}L^{n-1}U+...+p_{n-1}LU^{n-1}+p_{n}U^{n}}%
{p_{n+1}L^{n-1}+p_{n+2}L^{n-2}U+...+p_{2n-1}LU^{n-2}+p_{2n}U^{n-1}}\\
U^{\prime} &  =U+\frac{x+q_{0}U^{n}+q_{1}U^{n-1}L+...+q_{n-1}UL^{n-1}+q_{n}L^{n}}%
{q_{n+1}U^{n-1}+q_{n+2}U^{n-2}L+...+q_{2n-1}UL^{n-2}+q_{2n}L^{n-1}}%
\end{align*}
such that%
\[
L\leq L^{\prime}\leq\sqrt[n]{x}\leq U^{\prime}\leq U,
\]
which we will call \emph{contracting $n^{th}$ degree} maps. By choosing 
values
for the parameters $p=(p_{0},\ldots,p_{2n})$ and
$q=(q_{0},\ldots,q_{2n})$,
we
get each member of the family. For instance, the Secant-Newton map can be
obtained
by setting 
\begin{align*}
 p & =(-1,\underbrace{0,0,...,0}_\text{$n$},\underbrace{1,1,...,1}_\text{$n$})\\
 q & =(-1,\underbrace{0,0,...,0}_\text{$n$},n,\underbrace{0,0,...,0}_\text{$n$-1}). 
\end{align*}

The main contribution of this paper is the finding that the Secant-Newton
map is \emph{optimal} among all the contracting $n^{th}$ degree maps. By optimal,
we mean that the output interval of the Secant-Newton map is always a
subset
of the output interval of any  other contracting $n^{th}$ degree map, and in fact is almost always a proper subset.

This result generalizes   the previous result on the {\em square\/} root computation by 
Erascu-Hong \cite{erascu_square_2013}, where the Secant-Newton
map was shown to be optimal among all the contracting  {\em quadratic\/}  maps.
The new contributions, beyond the straightforward adaptation of \cite{erascu_square_2013},  are as follows.
\begin{itemize}
\item 
 We found that the precise notion of the ``optimality''  for the square-root case in \cite{erascu_square_2013}  could {\em not\/} be extended straightforwardly to the $n^{th}$ root case.  It had to be modified in a subtle but crucial way. 
See Theorem~\ref{thm:mainThm} (b). 
\item 
Furthermore, we found that the proof techniques used in~\cite{erascu_square_2013}    could {\em not\/} be  straightforwardly extended.  
In fact, it turns out that only a small part of the proof technique could  be straightforwardly generalized. However, the rest of the proof could not be generalized. Thus, we developed several new proof techniques.
See Lemmas \ref{lem:LL}, \ref{lem:closeU}, \ref{lem:UL},\ref{lem:(a)} and \ref{lem:(b)}.
\end{itemize}

The paper is structured as follows. In Section~\ref{sec:MainResult}, we
precisely state the main claim of the paper. In Section~\ref{sec:Proof},
we
prove the main claim.

\section{Main Result}

\label{sec:MainResult}In this section, we will make a precise statement of
the main result informally described in the previous section. For this, we
recall a few notions and notations.

\begin{definition}[$n^{th}$ degree map]
{\label{def:map} We say that a map 
\begin{equation*}
R:\;[L,U],x\;\;\mapsto \;\;[L^{\prime },U^{\prime }]
\end{equation*}%
is an \emph{$n^{th}$ degree map} if it has the following form
\begin{align*}
L^{\prime }& =L+\frac{x+p_{0}L^{n}+p_{1}L^{n-1}U+\cdots +p_{n}U^{n}}{%
p_{n+1}L^{n-1}+p_{n+2}L^{n-2}U+\cdots +p_{2n}U^{n-1}} \\
U^{\prime }& =U+\frac{x+q_{0}U^{n}+q_{1}U^{n-1}L+\cdots +q_{n}L^{n}}{%
q_{n+1}U^{n-1}+q_{n+2}U^{n-2}L+\cdots +q_{2n}L^{n-1}}.
\end{align*}%
We will denote such a map by $R_{p,q}$. }
\end{definition}

\begin{definition}[Secant-Newton map]
The \emph{Secant-Newton map} is the $n^{th}$ degree map $R_{p^{\ast },q^{\ast }}$
where $p^{\ast }=(-1,\underbrace{0,0,...,0}_{\text{n}},\underbrace{1,1,...,1}%
_{\text{n}})$ and $q^{\ast }=(-1,\underbrace{0,0,...,0}_{\text{n}},n,%
\underbrace{0,0,...,0}_{\text{n-1}})$, namely 
\begin{equation*}
R_{p^{\ast },q^{\ast }}:\;[L,U],x\;\;\mapsto \;\;[L^{\ast },U^{\ast }]
\end{equation*}%
where 
\begin{align*}
L^{\ast }& =L+\frac{x-L^{n}}{L^{n-1}+L^{n-2}U+\cdots +U^{n-1}} \\
U^{\ast }& =U+\frac{x-U^{n}}{nU^{n-1}}.
\end{align*}
\end{definition}

\begin{definition}[Contracting $n^{th}$ degree map]
\label{def:contracting} We say that a map 
\begin{equation*}
R:\;[L,U],x\;\;\mapsto\;\;[L^{\prime},U^{\prime}]
\end{equation*}
is a \emph{contracting $n^{th}$ degree map} if it is an $n^{th}$ degree map
and 
\begin{equation}
\underset{L,U,x}{\forall}\ \ 0<L\leq\sqrt[n]{x}\leq U\ \Longrightarrow\
L\leq L^{\prime}\leq\sqrt[n]{x}\leq U^{\prime}\leq U.  \label{eq:contraction}
\end{equation}
\end{definition}

\noindent Now we are ready to state the main result of the paper. %

\begin{theorem}[Main Result]
\label{thm:mainThm} Let $R_{p,q}$ be a contracting $\,$$n^{th}$ degree map
which is not $R_{p^{\ast},q^{\ast}}$ (Secant-Newton). Then we have

\begin{enumerate}
\item[(a)] {\small \textrm{$\;\;\underset{L,U,x}{\forall}\;\;\;\; 0<L\leq%
\sqrt[n]{x}\leq U \;\;\;\;\Longrightarrow\;\;\;\; R_{p^{*},q^{*}}([L,U],x)
\;\;\subseteq\;\; R_{p,q}([L,U],x) $ } }

\item[(b)] {\small \textrm{$\;\;\overset{o}{\underset{L,U,x}{\forall}}%
\;\;\;\;0<L\leq \sqrt[n]{x}\leq U\;\;\;\;\Longrightarrow
\;\;\;R_{p^{\ast},q^{\ast}}([L,U],x)\;\;\subsetneq\;\;R_{p,q}([L,U],x) $ } }
\end{enumerate}
\end{theorem}

\begin{remark}
The small circle above the universal quantifier in the second claim
indicates that the statement holds for almost all values of $L, U, x$.
Equivalently, this means the set of exceptions has measure zero.
\end{remark}

\begin{remark}
\textrm{The first claim states that the Secant-Newton map is never worse
than any other contracting $n^{th}$ degree map. The second claim states that
the Secant-Newton map is almost always better than all the other contracting 
$\,$$n^{th}$ degree maps.}
\end{remark}

\begin{remark}
This paper is a sequel to the square root case study of
Erascu-Hong \cite{erascu_square_2013}. Erascu-Hong proved the following two
results for the square-root case $\left( n=2\right) $: 
\begin{eqnarray*}
\underset{L,U,x}{\forall }\;\;\;\;0 &<&L\leq \sqrt[2]{x}\leq
U\;\;\;\;\Longrightarrow \;\;\;\;R_{p^{\ast },q^{\ast
}}([L,U],x)\;\;\subseteq \;\;R_{p,q}([L,U],x) \\
\underset{L,U,x}{\forall }\;\;\;\;0 &<&L<\sqrt[2]{x}<U\;\;\;\;%
\Longrightarrow \;\;\;\;R_{p^{\ast },q^{\ast }}([L,U],x)\;\;\subsetneq
\;\;R_{p,q}([L,U],x)
\end{eqnarray*}%
These look very similar to the two claims in Theorem \ref{thm:mainThm}
above for the $n\,^{th}$ root case (arbitrary $n$). In fact, the first
claims look identical to each other, which means that the claim for the $n=2$
generalizes to arbitrary $n$ without any change. On the other hand, the
second claims have subtle but important differences. 
\begin{equation*}
\begin{array}{lll}
n=2 & : & \underset{L,U,x}{\forall }\;\;\;\cdots \ \ L<\sqrt[2]{x}<U\ \cdots
\\ 
n=\text{arbitrary} & : & \overset{o}{\underset{L,U,x}{\forall }}\;\;\;\cdots
\ \ L\leq \sqrt[n]{x}\leq U\ \cdots%
\end{array}%
\end{equation*}%
Note that $\forall $ is replaced with $\overset{o}{\forall }$ and $<$ with $%
\ \leq .$ These subtle changes are necessary because, to our surprise, the
claim for $n=2$ does not hold in general. \noindent For example, consider
the following $n=3$ case: 
\[
p=(-1,0,0,0,2,\frac{1}{2},1), \ q=(1,0,0,0,3,0,0),\ L=1, \ U=2, \ x=\left(\frac{3}{2}\right)^{3}
\]  
This implies
$0<L<\sqrt[3]{x}<U$, and
\begin{align*}
&  R_{p^{\ast},q^{\ast}}([L,U],x)\;\;=\;\;R_{p,q}([L,U],x)\\
\iff &  L^{\prime}=L^{\ast}\;\wedge\;U^{\ast}=U^{\prime}\\
\iff &
\begin{array}
[t]{l}%
L+\frac{x-L^{3}}{p_{4}L^{2}+p_{5}LU+p_{6}U^{2}}=L+\frac{x-L^{3}}%
{L^{2}+LU+U^{2}}\\
\wedge\\
U+\frac{x-U^{3}}{3U^{2}}=U+\frac{x-U^{3}}{q_{4}U^{2}+q_{5}UL+q_{6}L^{2}}%
\end{array}
\ \ \
\begin{array}
[t]{l}%
\\
\text{{}}\\
\\
\end{array}
\\
\iff &
\begin{array}
[t]{l}%
p_{4}L^{2}+p_{5}LU+p_{6}U^{2}=L^{2}+LU+U^{2}\\
\wedge\\
q_{4}U^{2}+q_{5}UL+q_{6}L^{2}=3U^{2}%
\end{array}
\\
\iff &
\begin{array}
[t]{l}%
(p_{4}-1)L^{2}+(p_{5}-1)LU+(p_{6}-1)U^{2}=0\\
\wedge\\
(q_{4}-3)U^{2}+q_{5}UL+q_{6}L^{2}=0
\end{array}\\
\iff &
\begin{array}
[t]{l}%
(2-1)(1)^{2}+(\frac{1}{2}-1)(1)(2)+(1-1)(2)^{2}=0\\
\wedge\\
(3-3)(2)^{2}+(3)(1)(2)+(0)(1)^{2}=0
\end{array}\\
\iff &  0=0\;\;\;\;\wedge\;\;\;\;0=0\\
\iff &  true.
\end{align*}
\end{remark}

\section{Proof}
\label{sec:Proof} In this section, we will prove the main result (Theorem~%
\ref{thm:mainThm}). For the sake of easy readability, the proof will be
divided into several lemmas, which are interesting on their own. The main
theorem follows immediately from the last two lemmas (Lemmas \ref{lem:(a)}
and \ref{lem:(b)}).

This paper is a sequel to \cite{erascu_square_2013} where square-root ($n=2$) was considered. Hence we initially hoped that the proof techniques developed in \cite{erascu_square_2013} would be  generalizable straightforwardly to the $n^{th}$ root case.  
It turns out that a part of the proof could  indeed be straightforwardly generalized (Lemmas~\ref{lem:C1} and~\ref{lem:C2}). However, the rest of the proof could not be generalized at all. Thus, we had to develop several new proof techniques 
(Lemmas~\ref{lem:LL},~\ref{lem:closeU},~\ref{lem:UL},~\ref{lem:(a)},~\ref{lem:(b)}).

\begin{lemma}
\label{lem:C1}Let $R_{p,q}$ be a contracting $\,$$n^{th}$ degree map. Then
we have 
\begin{eqnarray*}
p_{0} &=&-1\ \wedge \ \ p_{1}=\cdots =p_{n}=0, \\
q_{0} &=&-1\ \wedge \ \ q_{1}=\cdots =q_{n}=0.
\end{eqnarray*}
\end{lemma}

\begin{proof}
Let $R_{p,q}$ be a contracting $\,n^{th}$ degree map. Then $p,q$ satisfy the
condition~\eqref{eq:contraction}. The proof essentially consists of
instantiating~the condition \eqref{eq:contraction} on $x=L^{n}$ and $%
x=U^{n}. $

By instantiating the condition \eqref{eq:contraction} with $x=L^{n}$ and
recalling the definition of $L^{\prime },$ we have 
\begin{equation*}
\underset{L,U}{\forall }\ \ 0<L\leq U\ \Longrightarrow \ L+\frac{%
L^{n}+p_{0}L^{n}+p_{1}L^{n-1}U+\cdots +p_{n}U^{n}}{%
p_{n+1}L^{n-1}+p_{n+2}L^{n-2}U+\cdots +p_{2n}U^{n-1}}=L.
\end{equation*}%
By simplifying, removing the denominator and collecting, we have 
\begin{equation*}
\underset{L,U}{\forall }\ \ \left( L,U\right) \in D\ \Longrightarrow g\left(
L,U\right) =0,
\end{equation*}%
where%
\begin{align*}
D& =\left\{ \left( L,U\right) :0<L\leq U\right\} , \\
g\left( L,U\right) & =\left( 1+p_{0}\right) L^{n}+p_{1}L^{n-1}U+\cdots
+p_{n}U^{n}.
\end{align*}%
Since the bivariate polynomial $g$ is zero over the 2-dim real domain $D$,
it must be identically zero. Thus its coefficients $1+p_{0},\ p_{1},\ldots
,p_{n}\ $ must be all zero.

By instantiating the condition \eqref{eq:contraction} with $x=U^{n}$ and
recalling the definition of $U^{\prime }$, we have 
\begin{equation*}
\underset{L,U}{\forall }\ \ 0<L\leq U\ \Longrightarrow \ U+\frac{%
U^{n}+q_{0}U^{n}+q_{1}U^{n-1}L+\cdots +q_{n}L^{n}}{%
q_{n+1}U^{n-1}+q_{n+2}U^{n-2}L+\cdots +q_{2n}L^{n-1}}=U.
\end{equation*}%
By simplifying, removing the denominator and collecting, we have 
\begin{equation*}
\underset{L,U}{\forall }\ \ \left( L,U\right) \in D\ \Longrightarrow \
h\left( L,U\right) =0,
\end{equation*}%
where 
\begin{align*}
D& =\left\{ \left( L,U\right) :0<L\leq U\right\} , \\
h\left( L,U\right) & =\left( 1+q_{0}\right) U^{n}+q_{1}U^{n-1}L+\cdots
+q_{n}L^{n}.
\end{align*}%
Since the bivariate polynomial $h$ is zero over the 2-dim real domain $D$, \
it must be identically zero. Thus its coefficients $1+q_{0},\ q_{1},\ldots
,q_{n}\ $ must be all zero.
\end{proof}

\begin{lemma}
\label{lem:C2} Let $R_{p,q}$ be a contracting $\,n^{th}$ degree map. Then we have 
\begin{align*}
L^{\prime }& =L+\frac{x-L^{n}}{p_{n+1}L^{n-1}+p_{n+2}L^{n-2}U+\cdots
+p_{2n}U^{n-1}} \\
U^{\prime }& =U+\frac{x-U^{n}}{q_{n+1}U^{n-1}+q_{n+2}U^{n-2}L+\cdots
+q_{2n}L^{n-1}}.
\end{align*}
\end{lemma}

\begin{proof}
Let $R_{p,q}$ be a contracting $\,n^{th}$ degree map. From Lemma \ref{lem:C1}, we
have 
\begin{eqnarray*}
p_{0} &=&-1\ \wedge \ \ p_{1}=\cdots =p_{n}=0 \\
q_{0} &=&-1\ \wedge \ \ q_{1}=\cdots =q_{n}=0
\end{eqnarray*}%
Recalling the definition of $L^{\prime }$ and $U^{\prime }$, we have%
\begin{align*}
L^{\prime }& =L+\frac{x-L^{n}}{p_{n+1}L^{n-1}+p_{n+2}L^{n-2}U+\cdots
+p_{2n}U^{n-1}} \\
U^{\prime }& =U+\frac{x-U^{n}}{q_{n+1}U^{n-1}+q_{n+2}U^{n-2}L+\cdots
+q_{2n}L^{n-1}}.
\end{align*}
\end{proof}

\begin{lemma}
\label{lem:LL} Let $R_{p,q}$ be a contracting $n^{th}$ degree map. Then we
have 
\begin{equation*}
\underset{L,U,x}{\forall }\;\;\;\;0<L\leq \sqrt[n]{x}\leq
U\;\;\;\;\Longrightarrow \;\;\;\;L^{\prime }\leq L^{\ast }.
\end{equation*}
\end{lemma}

\begin{proof}
Let $R_{p,q}$ be a contracting $n^{th}$ degree map. Then we have%
\begin{equation*}
\underset{L,U,x}{\forall }\ \ 0<L\leq \sqrt[n]{x}\leq U\ \Longrightarrow \
L^{\prime }\leq \sqrt[n]{x}.
\end{equation*}%
From Lemma $\ref{lem:C2},$ we have%
\begin{eqnarray*}
\underset{L,U,x}{\forall }\ \ 0 &<&L\leq \sqrt[n]{x}\leq U\ \Longrightarrow
\ L+\frac{x-L^{n}}{p_{n+1}L^{n-1}+p_{n+2}L^{n-2}U+\cdots +p_{2n}U^{n-1}}\leq 
\sqrt[n]{x} \\
\underset{L,U,x}{\forall }\ \ 0 &<&L\leq \sqrt[n]{x}\leq U\ \Longrightarrow
\ \frac{x-L^{n}}{p_{n+1}L^{n-1}+p_{n+2}L^{n-2}U+\cdots +p_{2n}U^{n-1}}\leq 
\sqrt[n]{x}-L \\
\underset{L,U,x}{\forall }\ \ 0 &<&L\leq \sqrt[n]{x}\leq U\ \Longrightarrow
\ \frac{\left( \sqrt[n]{x}-L\right) \left( \sqrt[n]{x}^{n-1}+\sqrt[n]{x}%
^{n-2}L+\cdots +L^{n-1}\right) }{p_{n+1}L^{n-1}+p_{n+2}L^{n-2}U+\cdots
+p_{2n}U^{n-1}}\leq \sqrt[n]{x}-L
\end{eqnarray*}%
By considering only the case $L<\sqrt[n]{x},$ we have%
\begin{equation*}
\underset{L,U,x}{\forall }\ \ 0<L<\sqrt[n]{x}\leq U\ \Longrightarrow \ \frac{%
\sqrt[n]{x}^{n-1}+\sqrt[n]{x}^{n-2}L+\cdots +L^{n-1}}{%
p_{n+1}L^{n-1}+p_{n+2}L^{n-2}U+\cdots +p_{2n}U^{n-1}}\leq 1
\end{equation*}%
Since $\sqrt[n]{x}^{n-1}+\sqrt[n]{x}^{n-2}L+\cdots +L^{n-1}>0\ $for $\ 0<L<%
\sqrt[n]{x},$ we have%
\begin{equation*}
\underset{L,U,x}{\forall }\ \ 0<L<\sqrt[n]{x}\leq U\ \Longrightarrow \ \frac{%
1}{p_{n+1}L^{n-1}+p_{n+2}L^{n-2}U+\cdots +p_{2n}U^{n-1}}\leq \frac{1}{\sqrt[n%
]{x}^{n-1}+\sqrt[n]{x}^{n-2}L+\cdots +L^{n-1}}
\end{equation*}%
By considering only the case $\sqrt[n]{x}=U,$ we have%
\begin{equation*}
\underset{L,U}{\forall }\ \ 0<L<U\ \Longrightarrow \ \frac{1}{%
p_{n+1}L^{n-1}+p_{n+2}L^{n-2}U+\cdots +p_{2n}U^{n-1}}\leq \frac{1}{%
U^{n-1}+U^{n-2}L+\cdots +L^{n-1}}
\end{equation*}
Since $x-L^{n}\geq 0$ for $L\leq \sqrt[n]{x},$ we have%
\begin{equation*}
\underset{L,U,x}{\forall }\;\;\;\;0<L\leq \sqrt[n]{x}\leq U\;\wedge\ L<U\
\Longrightarrow \ \frac{x-L^{n}}{p_{n+1}L^{n-1}+p_{n+2}L^{n-2}U+\cdots
+p_{2n}U^{n-1}}\leq \frac{x-L^{n}}{U^{n-1}+U^{n-2}L+\cdots +L^{n-1}}
\end{equation*}%
Since $x-L^{n}=0$ when $L=U,$ we have%
\begin{equation*}
\underset{L,U,x}{\forall }\;\;\;\;0<L\leq \sqrt[n]{x}\leq U\;\Longrightarrow
\ \frac{x-L^{n}}{p_{n+1}L^{n-1}+p_{n+2}L^{n-2}U+\cdots +p_{2n}U^{n-1}}\leq 
\frac{x-L^{n}}{U^{n-1}+U^{n-2}L+\cdots +L^{n-1}}
\end{equation*}%
By adding $L$ on both sides, we have%
\begin{equation*}
\underset{L,U,x}{\forall }\;\;\;\;0<L\leq \sqrt[n]{x}\leq U\ \Longrightarrow
L+\ \frac{x-L^{n}}{p_{n+1}L^{n-1}+p_{n+2}L^{n-2}U+\cdots +p_{2n}U^{n-1}}\leq
L+\frac{x-L^{n}}{U^{n-1}+U^{n-2}L+\cdots +L^{n-1}}
\end{equation*}%
Thus%
\begin{equation*}
\underset{L,U,x}{\forall }\;\;\;\;\;\;\;0<L\leq \sqrt[n]{x}\leq
U\Longrightarrow L^{\prime }\leq L^{\ast }.
\end{equation*}
\end{proof}

\begin{lemma}
\label{lem:closeU}
If\[ 
\underset{L,U,x}{\forall }\;\;\;\; 0<L\leq \sqrt[n]{x}<U\;\Longrightarrow A\geq B
\]
then 
\[
\underset{L,U}{\forall }\;\;\;\;0 <L<U\;\Longrightarrow C\geq B
\]
where 
\begin{align*}
A & = \frac{1}{U^{n-1}+U^{n-2}\sqrt[n]{x}+\cdots +\sqrt[n]{x}^{n-1}},  \\
B &=\frac{1}{%
q_{n+1}U^{n-1}+q_{n+2}U^{n-2}L+\cdots +q_{2n}L^{n-1}}, \\
C & =\frac{1}{nU^{n-1}}. 
\end{align*}
\end{lemma}
\begin{proof}
Assume
\begin{equation}
\underset{L,U,x}{\forall }\;\;\;\; 0<L\leq \sqrt[n]{x}<U\;\Longrightarrow A\geq B.
\label{eq:assume}
\end{equation}
We need to show 
\[
\underset{L,U}{\forall }\;\;\;\;0 <L<U\;\Longrightarrow C\geq B.
\]
Let $L, U$ be arbitrary but fixed such that $0<  L < U$. We need to prove that $C\geq~B$.
We will prove by contradiction, and thus assume $C <B$.
In order to derive a contradiction, we will try to find $\sqrt[n]{x}$ such that $0<L\leq \sqrt[n]{x}<U$ is true   but $A \geq B$ is false, which contradicts the assumption~\eqref{eq:assume}. 

Let $B-C=d$.  If $A-C <d$ then  $A \geq B$ is false. Thus it suffices to  find $\sqrt[n]{x}$ such that $0<L\leq \sqrt[n]{x}<U$ and $A - C < d$, that is,  $f(\sqrt[n]{x}) <0$ where
\begin{equation*}
f(z)=\frac{1}{U^{n-1}+U^{n-2}z+\cdots +z^{n-1}}-\frac{1}{nU^{n-1}}-d.
\end{equation*}

\noindent We consider two cases:
\begin{glists}{1em}{4em}{0em}{0.2em}
\iteml{Case 1:} $f(L) < 0$. Let $\sqrt[n]{x}=L$. Trivially $f(\sqrt[n]{x}) <0$.
\iteml{Case 2:} $f(L) \geq 0$.   
It is obvious that $f$ is continuous and monotonically decreasing over   $[L,U]$.
It is also obvious that $f(U) = -d<0$. Hence there exists a unique real root $\alpha$ of $f$ in $[L,U)$. Let $\sqrt[n]{x}=\frac{U+\alpha}{2}$. Then clearly $f(\sqrt[n]{x}) <0$. 
\end{glists}
Thus we have derived the desired contradiction in both cases. Hence $C \geq B$ and the lemma is proved. 
\end{proof}
\newpage

\begin{lemma}
\label{lem:UL} Let $R_{p,q}$ be a contracting $n^{th}$ degree map. Then we
have 
\begin{equation*}
\underset{L,U,x}{\forall }\;\;\;\; 0<L\leq \sqrt[n]{x}\leq
U\;\;\;\;\Longrightarrow \;\;\;\;U^{\ast }\leq U^{\prime }.
\end{equation*}
\end{lemma}

\begin{proof}
Let $R_{p,q}$ be a contracting $n^{th}$ degree map. Then we have%
\begin{equation*}
\underset{L,U,x}{\forall }\ \ 0<L\leq \sqrt[n]{x}\leq U\ \Longrightarrow \ 
\sqrt[n]{x}\leq U^{\prime }.
\end{equation*}%
From Lemma $\ref{lem:C2},$ we have%
\begin{eqnarray*}
\underset{L,U,x}{\forall }\;\;\;\;0 &<&L\leq \sqrt[n]{x}\leq
U\;\Longrightarrow \ \sqrt[n]{x}\leq U+\frac{x-U^{n}}{%
q_{n+1}U^{n-1}+q_{n+2}U^{n-2}L+\cdots +q_{2n}L^{n-1}} \\
\underset{L,U,x}{\forall }\;\;\;\;0 &<&L\leq \sqrt[n]{x}\leq
U\;\Longrightarrow \ U-\sqrt[n]{x}\geq \frac{U^{n}-x}{%
q_{n+1}U^{n-1}+q_{n+2}U^{n-2}L+\cdots +q_{2n}L^{n-1}} \\
\underset{L,U,x}{\forall }\;\;\;\;0 &<&L\leq \sqrt[n]{x}\leq
U\;\Longrightarrow \ U-\sqrt[n]{x}\geq \frac{\ \left( U-\sqrt[n]{x}\right)
\left( U^{n-1}+U^{n-2}\sqrt[n]{x}+\cdots +\sqrt[n]{x}^{n-1}\right) }{%
q_{n+1}U^{n-1}+q_{n+2}U^{n-2}L+\cdots +q_{2n}L^{n-1}}
\end{eqnarray*}%
By considering only the case $\sqrt[n]{x}<U$, we have%
\begin{equation*}
\underset{L,U,x}{\forall }\;\;\;\;0<L\leq \sqrt[n]{x}<U\;\Longrightarrow
1\geq \frac{U^{n-1}+U^{n-2}\sqrt[n]{x}+\cdots +\sqrt[n]{x}^{n-1}}{%
q_{n+1}U^{n-1}+q_{n+2}U^{n-2}L+\cdots +q_{2n}L^{n-1}}
\end{equation*}%
Since $U^{n-1}+U^{n-2}\sqrt[n]{x}+\cdots +\sqrt[n]{x}^{n-1}>0\ $for $0<L\leq 
\sqrt[n]{x}<U,$ we have%
\begin{equation*}
\underset{L,U,x}{\forall }\;\;\;\;0<L\leq \sqrt[n]{x}<U\;\Longrightarrow 
\frac{1}{U^{n-1}+U^{n-2}\sqrt[n]{x}+\cdots +\sqrt[n]{x}^{n-1}}\geq \frac{1}{%
q_{n+1}U^{n-1}+q_{n+2}U^{n-2}L+\cdots +q_{2n}L^{n-1}}
\end{equation*}%
By Lemma \ref{lem:closeU}
we have%
\begin{eqnarray*}
\underset{L,U}{\forall }\;\;\;\;0 <L<U\;\Longrightarrow \frac{1}{nU^{n-1}%
}\geq \frac{1}{q_{n+1}U^{n-1}+q_{n+2}U^{n-2}L+\cdots +q_{2n}L^{n-1}}
\end{eqnarray*}%
Since $x-U^{n}\leq 0$ for $\sqrt[n]{x}\leq U,$ we have%
\begin{equation*}
\underset{L,U,x}{\forall }\;\;\;\;0<L\leq \sqrt[n]{x}\leq U\;\wedge\
L<U\Longrightarrow \ \frac{x-U^{n}}{nU^{n-1}}\leq \ \frac{x-U^{n}}{%
q_{n+1}U^{n-1}+q_{n+2}U^{n-2}L+\cdots +q_{2n}L^{n-1}}
\end{equation*}%
Since $x-U^{n}=0$ when $L=U,$ we have%
\begin{equation*}
\underset{L,U,x}{\forall }\;\;\;\;0<L\leq \sqrt[n]{x}\leq U\;\Longrightarrow
\ \frac{x-U^{n}}{nU^{n-1}}\leq \ \frac{x-U^{n}}{%
q_{n+1}U^{n-1}+q_{n+2}U^{n-2}L+\cdots +q_{2n}L^{n-1}}
\end{equation*}%
By adding $U$ on both sides, we have%
\begin{equation*}
\underset{L,U,x}{\forall }\;\;\;\;0<L\leq \sqrt[n]{x}\leq U\ \Longrightarrow
U+\ \ \frac{x-U^{n}}{nU^{n-1}}\leq U+\frac{x-U^{n}}{%
q_{n+1}U^{n-1}+q_{n+2}U^{n-2}L+\cdots +q_{2n}L^{n-1}}
\end{equation*}%
Thus%
\begin{equation*}
\underset{L,U,x}{\forall }\;\;\;\;\;\;\;0<L\leq \sqrt[n]{x}\leq
U\Longrightarrow \;U^{\ast }\leq U^{\prime }.
\end{equation*}
\end{proof}

Now we are ready to prove the two claims in Main Theorem. The following
lemma (Lemma~\ref{lem:(a)}) will prove the claim (a) and the subsequent
lemma (Lemma~\ref{lem:(b)}) will prove the claim (b). %

\begin{lemma}[Main Theorem (a)]
\label{lem:(a)} Let $R_{p,q}$ be a contracting $n^{th}$ degree map which is
not $R_{p^{\ast},q^{\ast}}$ (Secant-Newton). Then we have 
\begin{equation*}
\underset{L,U,x}{\forall}\;\;\;\;0<L\leq\sqrt[n]{x}\leq
U\;\;\;\;\Longrightarrow\;\;\;\;R_{p^{\ast},q^{\ast}}([L,U],x)\;\;\subseteq
\;\;R_{p,q}([L,U],x).
\end{equation*}
\end{lemma}

\begin{proof}
Let $R_{p,q}$ be a contracting $n^{th}$ degree map which is not $R_{p^{\ast
},q^{\ast }}$ (Secant-Newton), that is, $p\neq p^{\ast }$ or $q\neq q^{\ast
} $. Let $L,U,x$ be arbitrary but fixed such that $0<L\leq \sqrt[n]{x}\leq
U. $ From Lemmas \ $\ref{lem:LL}$\ and $\ref{lem:UL},$ we have%
\begin{equation*}
\ L^{\prime }\leq L^{\ast }\;\wedge \;U^{\ast }\leq U^{\prime }.
\end{equation*}%
Hence $\ R_{p^{\ast },q^{\ast }}([L,U],x)\;\;\subseteq \;\;R_{p,q}([L,U],x).$
Main Theorem (a) has been proved.
\end{proof}

\begin{lemma}[Main Theorem (b)]
\label{lem:(b)} Let $R_{p,q}$ be a contracting $\,$$n^{th}$ degree map which
is not $R_{p^{\ast},q^{\ast}}$ (Secant-Newton). Then we have%
\begin{equation*}
\overset{o}{\underset{L,U,x}{\forall}}\;\;\;\;0<L \leq \sqrt[n]{x}\leq
U\;\;\;\;\Longrightarrow
\;\;\;R_{p^{\ast},q^{\ast}}([L,U],x)\;\;\subsetneq\;\;R_{p,q}([L,U],x).
\end{equation*}
\end{lemma}

\begin{proof}
Let $R_{p,q}$ be a contracting $n^{th}$ degree map which is not $R_{p^{\ast
},q^{\ast }}$ (Secant-Newton), that is, $p\neq p^{\ast }$ or $q\neq q^{\ast
} $. We need to show 
\begin{equation*}
\overset{o}{\underset{L,U,x}{\forall }}\;\;\;\;0<L\leq \sqrt[n]{x}\leq
U\;\;\;\;\Longrightarrow \;\;\;R_{p^{\ast },q^{\ast
}}([L,U],x)\;\;\subsetneq \;\;R_{p,q}([L,U],x).\newline
\end{equation*}%
It suffices to find a non-zero polynomial $f$\ in the variables $L,U,x\ $%
such that 
\begin{equation*}
\underset{f\left( L,U,x\right) \neq 0}{\underset{L,U,x}{\forall }}%
\;\;0<L\leq \sqrt[n]{x}\leq U\;\;\;\;\Longrightarrow \;\;\;R_{p^{\ast
},q^{\ast }}([L,U],x)\;\;\subsetneq \;\;R_{p,q}([L,U],x)\newline
\end{equation*}%
since the solution set of $f\left( L,U,x\right) =0\ $has measure zero. We
consider two cases:

\bigskip \noindent \textbf{Case 1}: $p\neq p^{\ast }.$ Let 
\begin{equation*}
f=\left( x-L^{n}\right) \left( (p_{n+1}-1)L^{n-1}+(p_{n+2}-1)L^{n-2}U+\cdots
+(p_{2n}-1)U^{n-1}\right). 
\end{equation*}%
Note that $f$ is a non-zero polynomial. Let $L,U,x$ be such that $f\left(
L,U,x\right) \neq 0$ and $0<L\leq \sqrt[n]{x}\leq U.$ We need to show that $%
R_{p^{\ast },q^{\ast }}([L,U],x)\;\;\subsetneq \;\;R_{p,q}([L,U],x).$ From
Lemma \ref{lem:(a)}, we already have $R_{p^{\ast },q^{\ast
}}([L,U],x)\;\subseteq \;\;R_{p,q}([L,U],x).$ Thus it suffices to show that 
\begin{equation*}
R_{p^{\ast },q^{\ast }}([L,U],x)\;\;\neq \;\;R_{p,q}([L,U],x).\newline
\end{equation*}%
Note%
\begin{eqnarray*}
&&f\left( L,U,x\right) \neq 0 \\
&\Longrightarrow &\ \left( x-L^{n}\right) \left(
(p_{n+1}-1)L^{n-1}+(p_{n+2}-1)L^{n-2}U+\cdots +(p_{2n}-1)U^{n-1}\right) \neq
0 \\
&\Longrightarrow &\left( x-L^{n}\right) \left( \left(
p_{n+1}L^{n-1}+p_{n+2}L^{n-2}U+\cdots +p_{2n}U^{n-1}\right) -\left(
L^{n-1}+L^{n-2}U+\cdots +U^{n-1}\right) \right) \neq 0 \\
&\Longrightarrow &\frac{x-L^{n}}{p_{n+1}L^{n-1}+p_{n+2}L^{n-2}U+\cdots
+p_{2n}U^{n-1}}\neq \frac{x-L^{n}}{L^{n-1}+L^{n-2}U+\cdots +U^{n-1}} \\
&\Longrightarrow &L^{\prime }\neq L^{\ast } \\
&\Longrightarrow &R_{p^{\ast },q^{\ast }}([L,U],x)\;\;\neq
\;\;R_{p,q}([L,U],x).\newline
\end{eqnarray*}

\bigskip \noindent \textbf{Case 2}: $q\neq q^{\ast }.$ Let 
\begin{equation*}
f=\left( x-U^{n}\right) \left( \left( q_{n+1}-n\right)
U^{n-1}+q_{n+2}U^{n-2}L+\cdots +q_{2n}L^{n-1}\right).
\end{equation*}%
Note that $f$ is a non-zero polynomial. Let $L,U,x$ be such that $f\left(
L,U,x\right) \neq 0$ and $0<L\leq \sqrt[n]{x}\leq U.$ We need to show that $%
R_{p^{\ast },q^{\ast }}([L,U],x)\;\;\subsetneq \;\;R_{p,q}([L,U],x).$ From
Lemma \ref{lem:(a)}, we already have $R_{p^{\ast },q^{\ast
}}([L,U],x)\;\subseteq \;\;R_{p,q}([L,U],x).$ Thus it suffices to show that 
\begin{equation*}
R_{p^{\ast },q^{\ast }}([L,U],x)\;\;\neq \;\;R_{p,q}([L,U],x).\newline
\end{equation*}%
Note%
\begin{eqnarray*}
&&f\left( L,U,x\right) \neq 0 \\
&\Longrightarrow &\ \left( x-U^{n}\right) \left( \left( q_{n+1}-n\right)
U^{n-1}+q_{n+2}U^{n-2}L+\cdots +q_{2n}L^{n-1}\right) \neq 0 \\
&\Longrightarrow &\left( x-U^{n}\right) \left( \left(
q_{n+1}U^{n-1}+q_{n+2}U^{n-2}L+\cdots +q_{2n}L^{n-1}\right) -nU^{n-1}\right)
\neq 0 \\
&\Longrightarrow &\frac{x-U^{n}}{q_{n+1}U^{n-1}+q_{n+2}U^{n-2}L+\cdots
+q_{2n}L^{n-1}}\neq \frac{x-U^{n}}{nU^{n-1}} \\
&\Longrightarrow &U^{\prime }\neq U^{\ast } \\
&\Longrightarrow &R_{p^{\ast },q^{\ast }}([L,U],x)\;\;\neq
\;\;R_{p,q}([L,U],x).\newline
\end{eqnarray*}

\noindent Main Theorem (b) has been proved.
\end{proof}

\section{Conclusion}
In this paper we extended a previous work on the optimal square root computation by 
Erascu-Hong \cite{erascu_square_2013} to arbitrary   $n^{th}$  root
computation.  The contributions are as follows.
\begin{itemize}
\item We proved that the well
known Secant-Newton refinement map is  ``optimal'' among  its natural
generalizations, that is, among the maps that are
contracting and are certain rational functions.
\item 
 We found that the precise notion of the ``optimality''  for the square-root case in \cite{erascu_square_2013}  could {\em not\/} be extended straightforwardly to the $n^{th}$ root case.  It had to be modified in a subtle but crucial way.  
\item 
Furthermore, we found that the proof techniques used in~\cite{erascu_square_2013}    could {\em not\/} be  straightforwardly extended.  
In fact, it turns out that only a small part of the proof technique could  be straightforwardly generalized. However, the rest of the proof could not be generalized. Thus, we developed several new proof techniques.
\end{itemize}
This work motivates several interesting further questions.
\begin{itemize}
\item What about dropping the condition ``contracting''?
The Secant-Newton map is a particular instance of interval
Newton map with slope where $m$ is chosen to be $U$.
If one chooses a different $m$ value (from $U$), then the interval
Newton map with slope
is not contracting. In practice, one remedies this by intersecting
the result of the map
with $[L,U]$ before the next iteration. This trivially ensures
that the resulting map is contracting. This motivates a larger
family
of maps where a map is defined as a quadratic map   composed with
intersection with $[L,U]$.
One asks what  the optimal map is among the larger family of maps.\item  What about   broadening the scope of this problem? The problem tackled in this paper could be recast as follows: given  a positive number $x$, find the positive  real root of the polynomial equation $f(y)=y^n-x$, using an interval refinement map. This motivates the following natural generalization: find a real root of an arbitrary polynomial equation in a given  interval, using an interval refinement map.  Again, one could ask what  the optimal refinement map is among a naturally chosen  family of maps.
 \end{itemize}
We leave them as open problems/challenges for future research.

\bibliographystyle{plain}

\end{document}